\documentclass[11pt]{article}
\usepackage{float}
\usepackage{algorithm}
\usepackage[noend]{algpseudocode}
\usepackage{amsmath, amsthm, amssymb, amsfonts}
\usepackage{graphicx}
\usepackage{color}
\usepackage{xy}
\usepackage{fullpage}
\usepackage{ifthen}
\usepackage{enumitem}
\usepackage{comment}

\newcommand{\cC}{\ensuremath{\mathcal{C}}}

\newcommand{\cE}{\ensuremath{\mathcal{E}}}

\newcommand{\cI}{\ensuremath{\mathcal{I}}}

\newcommand{\cK}{\ensuremath{\mathcal{K}}}
\newcommand{\cL}{\ensuremath{\mathcal{L}}}
\newcommand{\cM}{\ensuremath{\mathcal{M}}}
\newcommand{\cN}{\ensuremath{\mathcal{N}}}
\newcommand{\cO}{\ensuremath{\mathcal{O}}}

\newcommand{\cP}{\ensuremath{\mathcal{P}}}

\newcommand{\cY}{\ensuremath{\mathcal{Y}}}
\newcommand{\cX}{\ensuremath{\mathcal{X}}}

\newcommand{\bbB}{\ensuremath{\mathbb{B}}}

\newcommand{\bbG}{\ensuremath{\mathbb{G}}}

\newcommand{\bbP}{\ensuremath{\mathbb{P}}}
\newcommand{\bbQ}{\ensuremath{\mathbb{Q}}}

\newcommand{\bbS}{\ensuremath{\mathbb{S}}}

\DeclareMathOperator{\Bd}{\partial}

\DeclareMathOperator{\Faces}{Faces}

\DeclareMathOperator{\codim}{codim}

\DeclareMathOperator{\inte}{Int}

\DeclareMathOperator{\names}{names}

\DeclareMathOperator{\skel}{skel}

\DeclareMathOperator{\views}{views}
\DeclareMathOperator{\view}{view}
\newcommand{\var}[1]{\lstinline+#1+}

\theoremstyle{remark}
\newtheorem{theorem}{Theorem}[section]

\newtheorem{corollary}[theorem]{Corollary}
\newtheorem{definition}[theorem]{Definition}

\newtheorem{lemma}[theorem]{Lemma}

\newcommand{\set}[1]{\{#1\}} 
\newcommand{\send}[1]{\mathbf{send}({#1})}
\newcommand{\recv}[1]{\mathbf{recv}({#1})}

\newcommand{\runmap}[1]{{\, \underrightarrow{\quad {#1} \quad} \,}}
\newcommand{\piece}[2]{[{#1}:{#2}]}
\newcommand{\quorum}[1]{{\textrm{Quorum}({#1})}}

\newcommand{\setproc}[1]{\textrm{SetProc}({#1})}
\newcommand{\setcons}[1]{\textrm{SetCons}({#1})}

\DeclareMathOperator{\cont}{Cont}
\DeclareMathOperator{\interp}{Interp}

\algblockdefx[upon]{Upon}{EndUpon}%
	[1]{\textbf{upon} #1 \textbf{do}}%
	{}

\title{Tight Bounds for Connectivity and Set Agreement\\in Byzantine Synchronous Systems\\
~\\
{\large Regular Submission}}
\author{Hammurabi Mendes (Davidson College), Maurice Herlihy (Brown University)}
\date{}

\begin{document}

\maketitle

\begin{abstract}
In this paper,
we show that the protocol complex of a Byzantine synchronous system
can remain $(k - 1)$-connected for up to $\lceil t/k \rceil$ rounds,
where $t$ is the maximum number of Byzantine processes,
and $t \ge k \ge 1$.
This topological property implies that
$\lceil t/k \rceil + 1$ rounds are necessary to solve $k$-set agreement
in Byzantine synchronous systems,
compared to  $\lfloor t/k \rfloor + 1$ rounds in synchronous crash-failure systems.
We also show that our connectivity bound is tight
as we indicate solutions to Byzantine $k$-set agreement in exactly
$\lceil t/k \rceil + 1$ synchronous rounds,
at least when $n$ is suitably large compared to $t$.
In conclusion,
we see how Byzantine failures can potentially require \emph{one} extra round
to solve $k$-set agreement, and,
for $n$ suitably large compared to $t$, \emph{at most that}.
\end{abstract}

\section{Introduction}
\label{Sec-Introduction}

A \emph{task} is a distributed coordination problem
where multiple processes start with private inputs,
communicate among themselves (by shared memory or message passing),
and halt with outputs consistent with the task specification.
There are \emph{crash-failure} systems~\cite{DSBook},
where processes can fail only by permanent, unannounced halting,
or \emph{Byzantine-failure} systems~\cite{Lamport1982},
where processes can fail arbitrarily, even maliciously.
In \emph{synchronous} systems,
communication and computation are organized in discrete rounds.
In each round, each non-faulty process performs as follows, in order: (i) sends a message;
(ii) receives all messages sent in the current round by the other processes;
and (iii) performs internal computation.
In \emph{asynchronous} systems,
processes may have different relative speeds,
and communication is subject to unbound, finite delays.

The problem of consensus in the synchronous Byzantine message-passing
model was among the earliest to be investigated,
and upper and lower consensus bounds in that model are well-understood.
In this paper,
we turn our attention to overall computational power of this model,
including bounds for problems such as $k$-set agreement.
We use concepts and techniques adapted from combinatorial topology.
In essence,
we can capture all possible information dissemination patterns
permitted by this model in a single combinatorial structure called a
\emph{simplicial complex} (or just \emph{complex}).
A classical topological property of a simplicial complex is its level
of \emph{connectivity}, which is, roughly speaking, the dimension below
which it has no holes.
Many classical proofs of consensus impossibility can be reformulated
as showing that certain complexes are 0-connected
(also called \emph{path-connected}),
and all known impossibility proofs for $k$-set agreement rely on
showing that certain complexes are $(k-1)$-connected.
Very informally, the higher the degree of connectivity imposed by the adversary,
the weaker the model's computational power.
Here, we present the first tight bounds on connectivity for the
synchronous Byzantine message-passing model.

Prior work using topological techniques is discussed in Sec.~\ref{Sec-RelatedWork}.
Our operational setting is detailed in Sec.~\ref{Sec-OperationalModel},
and our topological model is formalized in Sec.~\ref{Sec-TopologicalModel}.

Our \textbf{first contribution} comes in Sec.~\ref{Sec-ConnectivityUpperBound}.
We show that,
in a Byzantine synchronous system,
the protocol complex can remain \emph{$(k - 1)$-connected} for $\lceil t/k \rceil$ rounds,
where $t$ is an upper bound on the number of Byzantine processes.
Perhaps surprisingly,
this is only \emph{one} more round than the upper bound for crash-failure systems
($\lfloor t/k \rfloor$, shown in~\cite{ChaudhuriHLT2000}).
Technically,
we conceive a combinatorial operator modeling the ability of Byzantine processes
to \emph{equivocate} -- that is,
to transmit ambiguous state information --
without revealing their Byzantine nature.
We compose this operator with regular crash-failure operators,
extending the protocol complex connectivity for one extra round.
As noted, connectivity is of interest because a $(k - 1)$-connected
protocol complex prevents important problems such as $k$-set
agreement~\cite{ksetagreement,dePriscoMR2001} from having solutions.

Our \textbf{second contribution} comes in Sec.~\ref{Sec-KSetLowerBound}.
We show that the above connectivity bound
is \emph{tight} in certain settings (described in Sec.~\ref{Sec-KSetLowerBound}),
by solving $k$-set agreement in $\lceil t/k \rceil + 1$ rounds.
We do so with a full-information protocol that assumes $n$ suitably large compared to $t$.
The protocol suits well our purpose of tightening the $\lceil t/k \rceil$ bound,
and also exposes clearly \emph{the reason why} $\lceil t/k \rceil + 1$ rounds is enough
to solve $k$-set agreement.

These results give new insight into the power of Byzantine
adversaries for problems beyond consensus.
Although Byzantine adversaries seem much more powerful than
crash-failure ones,
we show that a Byzantine adversary can impose at most \emph{one} additional
synchronous round beyond that imposed by a crash-failure adversary.
In terms of solvability vs. number of rounds,
the penalty for moving from crash to Byzantine failures,
captured by $(k - 1)$-connectivity in the protocol complex,
can be \emph{quite limited} in synchronous systems,
particularly when $n$ is relatively large compared to $t$.

\section{Related Work}
\label{Sec-RelatedWork}

The Byzantine failure model was initially introduced by Lamport,
Shostak, and Pease~\cite{Lamport1982}.
The use of simplicial complexes to model distributed computations
was introduced by Herlihy and Shavit~\cite{HerlihyShavit1993}.
The asynchronous computability theorem for general tasks in~\cite{HerlihyShavit1999}
details the approach for asynchronous wait-free computation in the
crash-failure model.
This model was recently generalized by~Gafni, Kuznetsov, and Manolescu~\cite{GACT2014}.
Computability in Byzantine asynchronous systems,
where tasks are constrained in terms of non-faulty inputs,
was recently considered in~\cite{MendesHerlihy14}.

The $k$-set agreement problem was originally defined by Chaudhuri~\cite{ksetagreement}.
Alternative formulations with
different validity notions, or failure/communication settings,
are discussed in~\cite{Neiger93,dePriscoMR2001}.
A full characterization of optimal translations between different failure settings
is given in~\cite{BazziNeiger01,NeigerToueg90},
which requires different number of rounds depending on the relation between
the number of faulty processes,
and the number of participating processes.

The relationship between connectivity and the impossibility of $k$-set agreement
is described explicitly or implicitly in~\cite{ChaudhuriHLT2000,HerlihyShavit1999,SaksZ1993}.
Recent work by Castañeda, Gonczarowski, and Moses~\cite{CastaMosesBA2013}
considers an issue of chains of hidden values,
a concept loosely explored here.
The approach based on shellability and layered executions for lower bounds in connectivity
has been used by Herlihy, Rajsbaum, and Tuttle~\cite{HerlihyRT98,HerlihyRT09,ConcurrentShellable},
assuming crash-failure systems, synchronous or asynchronous.

\section{Operational Model}
\label{Sec-OperationalModel}

We have $n+1$ processes\footnote{
Choosing $n+1$ processes rather than $n$ simplifies the topological notation,
but slightly complicates the computing notation.
Choosing $n$ processes has the opposite trade-off.
We choose $n+1$ for compatibility with prior work.
}
$\bbP = \set{P_0, \ldots, P_n}$
communicating by message-passing via pairwise, reliable, FIFO channels
(\emph{authenticated channels} in the literature~\cite{RSPBook}).
Technically,
all transmitted messages are delivered uniquely, in FIFO order,
and with sender reliably identified.

At most $t$ processes are \emph{faulty} or \emph{Byzantine}~\cite{Lamport1982},
and may display arbitrary, even malicious behavior, at any point in the execution.
The actual behavior of Byzantine processes
is defined by an \emph{adversary}.
Byzantine processes may execute the protocol correctly or incorrectly,
at the discretion of the adversary.
Processes behaving in strict accordance to the protocol for rounds $1$ up to some $r$ (inclusive)
are called \emph{non-faulty processes up to round $r$},
and are denoted by $\bbG^r$.
A non-faulty process up to any round $r \ge 1$ is called simply \emph{non-faulty} or \emph{correct},
which we denote by $\bbG$.

We model processes as state machines.
The input value (resp. output value) of a non-faulty process $P_i$ is written $I_i$ (resp. $O_i$).
Byzantine processes may have ``apparent'' inputs,
denoted as above.
Each non-faulty process $P_i$ has an internal state called \emph{view}, which we denote by $\view(P_i)$.
In the beginning of the protocol, $\view(P_i)$ is $I_i$.
At any round $r$,
any non-faulty process:
(1) sends its internal state to all other processes;
(2) receives the state information from other processes;
(3) concatenates that information to its own internal state.
After completing some number of iterations,
each process applies a decision function $\delta$ to its current state in order to decide $O_i$.
Thus,
we assume that processes follow a \emph{full-information} protocol~\cite{HerlihyRT09}.

For simplicity of notation,
we define a round $0$ where processes are simply assigned their inputs.
Without losing generality,
all processes are assumed non-faulty up to round $0$:
$\bbG^0 = \bbP$ and $\bbB^0 = \emptyset$.
For any round $r \ge 0$,
a \emph{global state} formally specifies:
(1) the non-faulty processes up to round $r$;
and (2) the view of all non-faulty processes up to round $r$.

\section{Topological Model}
\label{Sec-TopologicalModel}

We now sketch the required concepts from combinatorial topology.
For details, please refer to Munkres~\cite{Munkres84}, Kozlov \cite{Kozlov07}, or Herlihy \emph{et al.}~\cite{MauriceBook}.



\textbf{Basics. }
A \emph{simplicial complex} $\cK$ consists of a finite set $V$
along with a collection of subsets of $V$ closed under containment.
An element of $V$ is called a \emph{vertex} of $\cK$.
The set of vertices of $\cK$ is referred by $V(\cK)$.
Each set in $\cK$ is called a \emph{simplex},
usually denoted by lower-case Greek letters: $\sigma, \tau$, etc.
The \emph{dimension} $\dim(\sigma)$ of a simplex $\sigma$ is $|\sigma|-1$.

A subset of a simplex is called a \emph{face}.
The collection of faces of $\sigma$ with dimension exactly $x$ is called $\Faces^x(\sigma)$.
A face $\tau$ of $\sigma$ is called \emph{proper} if $\dim(\tau) = \dim(\sigma) - 1$.
We use ``$k$-simplex'' as shorthand for ``$k$-dimensional simplex'',
also in ``$k$-face.''
The dimension $\dim(\cK)$ of a complex is the maximal dimension of its simplexes,
and a \emph{facet} of $\cK$ is any simplex having maximal dimension in $\cK$.
A complex is said \emph{pure} if all facets have dimension $\dim(\cK)$.
The set of simplexes of $\cK$
having dimension at most~$\ell$ is a subcomplex of $\cK$,
which is called \emph{$\ell$-skeleton} of $\cK$,
denoted by $\skel^\ell(\cK)$.
 

\textbf{Maps. }
Let $\cK$ and $\cL$ be complexes.
A \emph{vertex map} $f$ carries vertices of $\cK$ to vertices of $\cL$.
If $f$ additionally carries simplexes of $\cK$ to simplexes of $\cL$,
it is called a \emph{simplicial} map.
A~\emph{carrier map} $\Phi$ from $\cK$ to $\cL$ takes each
simplex $\sigma\in\cK$ to a subcomplex $\Phi(\sigma) \subseteq \cL$,
such that for all $\sigma,\tau \in \cK$,
we have $\Phi(\sigma \cap \tau) \subseteq \Phi(\sigma) \cap \Phi(\tau)$.
A simplicial map $\phi: \cK \to \cL$ is \emph{carried by the carrier map}
$\Phi: \cK \to 2^{\cL}$ if, for every simplex $\sigma \in \cK$,
we have $\phi(\sigma) \subseteq \Phi(\sigma)$.

Although we defined simplexes and complexes in a purely combinatorial way,
they can also be interpreted geometrically.
An $n$-simplex can be identified with the convex hull of $(n+1)$
affinely-independent points in the Euclidean space of appropriate dimension.
This geometric realization can be extended to complexes.
The point-set that underlies such \emph{geometric complex} $\cK$ is called
the \emph{polyhedron} of $\cK$, denoted by $|\cK|$.
For any simplex $\sigma$, the \emph{boundary} of $\sigma$, which we denote $\Bd{\sigma}$,
is the simplicial complex of $(\dim(\sigma) - 1)$-faces of $\sigma$.
The \emph{interior} of $\sigma$ is defined as
$\inte{\sigma} = |\sigma| \setminus |\Bd{\sigma}|$.

We can define simplicial/carrier maps between geometrical complexes.
Given a simplicial map $\phi: \cK \to \cL$ (resp. carrier map $\Phi: \cK \to 2^{\cL}$),
the polyhedrons of every simplex in $\cK$ and $\cL$
induce a continuous simplicial map $\phi_c: |\cK| \to |\cL|$ (resp. continuous carrier map $\Phi_c: |\cK| \to |2^{\cL}|$).
We say $\phi$ (resp. $\phi_c$) is carried by $\Phi$
if, for any $\sigma \in \cK$,
we have $|\phi(\sigma)| \subseteq |\Phi(\sigma)|$ (resp. $\phi_c(|\sigma|) \subseteq \Phi_c(|\sigma|)$).


\textbf{Connectivity.} 
In light of topology,
two geometrical objects $A$ and $B$ are \emph{homeomorphic} if,
there is a continuous map from $A$ into $B$ or vice-versa.
Technically,
there exists a continuous map between those objects, in either direction \cite{Munkres00,Munkres84}.
%
%
We say that a simplicial complex $\cK$
is \emph{$x$-connected}, $x \ge 0$,
if every continuous map
of a subset of $|\cK|$ homeomorphic to an $x$-sphere in $|\cK|$
can be extended into
a subset of $|\cK|$ homeomorphic to an $(x+1)$-disk in $|\cK|$.
In analogy,
think of the extremes of a pencil as a $0$-disk,
and the pencil itself as a $1$-sphere (the extension is possible if $0$-connected);
the rim of a coin as a $1$-sphere,
and the coin itself as a $2$-disk (the extension is possible if $1$-connected);
the outer layer of a billiard ball as a $2$-sphere,
and the billiard ball itself as a $3$-disk (the extension is possible if $2$-connected).
For us,
$(-1)$-connected is understood as \emph{non-empty},
and $(-2)$-connected or lower imposes no restriction.


\begin{definition}
\label{definition-pseudosphere}
Let $\bbS = \set{(P_i,S_i): P_i \in \bbP'}$,
where each $S_i$ is an arbitrary set and $\bbP' \subseteq \bbP$.
A \emph{pseudosphere} $\Psi(\bbP', \bbS)$ is a simplicial complex where
$\sigma \in \Psi(\bbP', \bbS)$ if
$\sigma = \set{(P_i,V_i): P_i \in \bbP', V_i \in S_i}$.
\end{definition}

Essentially,
a pseudosphere
is a simplicial complex formed by independently assigning values to all the specified processes.
If $S_i = S$ for all $P_i \in \bbP'$,
we simply write $\Psi(\bbP', S)$.


\begin{definition}
\label{definition-shellable}
A pure, simplicial complex $\cK$ is \emph{shellable} if we can arrange the facets of $\cK$
in a linear order $\phi_0 \ldots, \phi_t$ such that
$ \left( \bigcup_{0 \le i < k} \phi_i \right) \cap \phi_k $
is a pure $(\dim(\phi_k) - 1)$-dimensional simplicial complex for all $0 < k \le t$.
We call the above linear order $\phi_0, \ldots, \phi_t$ a \emph{shelling order}.
\end{definition}

Intuitively,
a simplicial complex is shellable
if it can be built by gluing its $x$-simplexes along their $(x - 1)$ faces only,
where $x$ is the dimension of the complex.
Note that $\phi_0, \ldots, \phi_t$ is a shelling order
if any $\phi_i \cap \phi_j$ ($0\le i < j \le t$)
is contained in a $(\dim(\phi_k) - 1)$-face of $\phi_k$ ($0\le k < j$).
Hence,
\begin{align}
\textrm{for any } i < j \textrm{ exists } k < j \textrm{ where } (\phi_i \cap \phi_j) \subseteq (\phi_k \cap \phi_j)
\textrm{ and }
|\phi_j \setminus \phi_k| = 1 \mathrm{.}
\end{align}
Shellability and pseudospheres are important tools to
characterize connectivity in simplicial complexes.
The following lemmas are proved
in~\cite{ConcurrentShellable} and~\cite{MauriceBook} (pp. 252--253).

\begin{lemma}
\label{lemma-pseudosphere-shellable}
Any pseudosphere $\phi(\bbP', \bbS)$ is shellable,
considering arbitrary $\bbS = \set{(P_i,S_i): \forall P_i \in \bbP'}$.
\end{lemma}

\begin{lemma}
\label{lemma-shellable-connected}
For any $k \ge 1$,
if the simplicial complex $\cK$ is shellable and $\dim(\cK) \ge k$
then $\cK$ is $(k - 1)$-connected.
\end{lemma}


\textbf{Nerve Theorem.} 
Let $\cK$ be a simplicial complex with a \emph{cover}
$\set{\cK_i: i \in I} = \cK$,
where $I$ is a finite index set.
The \emph{nerve} $\cN(\set{\cK_i: i \in I})$ is the simplicial complex with vertexes $I$
and simplexes $J \subseteq I$ whenever $ \cK_J = \bigcap_{j \in J} \cK_j \ne \emptyset \mathrm{.} $
We can characterize the connectivity of $\cK$ in terms of the connectivity of
the intuitively simpler nerve of $\cK$
with the next theorem.

\begin{theorem}[Nerve Theorem \cite{Kozlov07,Bjorner1995}]
\label{theorem-nerve}
If for any $J \subseteq I$ denoting a simplex of $\cN(\set{\cK_i: i \in I})$
(thus, $\cK_J \ne \emptyset$)
we have that $\cK_J$ is $(k - |J| + 1)$-connected,
then $\cK$ is $k$-connected if and only if $\cN(\set{\cK_i: i \in I})$ is $k$-connected.
\end{theorem}


\textbf{Protocol Complexes.} 
We represent the evolution of the global state of the system throughout the rounds
by simplicial complexes that we call \emph{protocol complexes}.

\begin{definition}
\label{definition-labeled}
For $r \ge 0$,
a \emph{name-view} simplex $\sigma$ is such that:
(i) $\sigma = \set{(P_i,\view^r(P_i)): \forall P_i \in \bbG^r}$,
where $\view^r(P_i)$ denotes $P_i$'s view at round $r$;
and (ii) if $(P_i,\view^r(P_i))$ and $(P_j,\view^r(P_j))$ are both in $\sigma$,
then $P_i \ne P_j$.
\end{definition}

Unless otherwise noted, all of our simplicial and carrier maps $f$ are such that $\names(\sigma) = \names(f(\sigma))$, that is, they map between vertices associated with the same processes.

\begin{definition}
For any name-view simplex $\sigma$, define
$\names(\sigma) = \set{P_i: \exists V \textrm{ such that } (P_i,V) \in \sigma}$
and
$\views(\sigma) = \set{V_i: \exists P \textrm{ such that } (P,V_i) \in \sigma}$.
\end{definition}

The round-$0$ protocol complex $\cK^0$ has
name-view $n$-simplexes $\sigma_I = \set{(P_i,I_i): \forall P_i \in \bbG^0}$,
representing all the possible process inputs in the beginning of the protocol.
The round-$r$ protocol complex $\cK^r$,
for any $r \ge 0$,
is defined as follows:
if $\sigma \in \cK^r$,
then $\sigma = \set{(P_i,\view^r(P_i)): \forall P_i \in \bbG^r}$,
representing a possible global state of the system for round $r$.

\section{Connectivity Upper Bound}
\label{Sec-ConnectivityUpperBound}
Informally, if the adversary displays Byzantine behavior early in the execution,
then in a synchronous, full-information protocol,
subsequent communication among the non-faulty processes can reveal
the identities of the Byzantine processes,
using simple techniques inspired from \cite{BazziNeiger01,Bracha,SriTouRB}.
Instead, it behooves the adversary to postpone malicious behavior to
the very last round, where it cannot detected.

Say that non-faulty processes start the computation with inputs in $V = \set{v_0, \ldots, v_d}$,
\emph{arbitrarily} assigned,
with some $d \ge k$ and $t \ge k \ge 1$.
To prove our upper bound,
we show how the adversary can impose a particular admissible execution
that preserves high connectivity in the protocol complex.

Let $r = \lfloor t/k \rfloor$ and $m = t \bmod k$.
We have $r$ \emph{crash rounds},
where in each round $k$ processes fail by crashing, but display no Byzantine behavior.
If $m > 0$,
we have an extra \emph{equivocation round},
where a single Byzantine process
sends different views to different processes,
causing extra confusion.
This round-by-round execution produces a sequence of protocol complexes
$\cK^0, \ldots, \cK^{r+1}$,
related by carrier maps
$\cC^i: \cK^{i-1} \to 2^{\cK^i}$,
for $1 \le i \le r$,
and $\cE: \cK^r \to 2^{\cK^{r+1}}$.
\begin{equation}
\label{equation-setting}
\cK^0 \runmap{\cC^1} \cK^1 \ldots \runmap{\cC^r} \cK^r
\underbrace{\runmap{\cE} \cK^{r+1}}_{\textrm{only if $m > 0$}} \mathrm{.}
\end{equation}


In each of the first $r$ rounds,
exactly $k$ processes are failed by the adversary.
The crash-failure carrier maps are defined
as follows~\cite{ConcurrentShellable,MauriceBook}:

\begin{definition}
\label{definition-roundop-crash}
For any $1 \le i \le r$,
the crash-failure operator $\cC^i: \cK^{i-1} \to 2^{\cK^i}$ is such that
\begin{equation}
\label{eq-roundop-crash}
\cC^i(\sigma) = \bigcup_{\tau \in \Faces^{n-ik}(\sigma)}
		\Psi(\names(\tau); \piece{\tau}{\sigma})
\end{equation}
for any $\sigma \in \cK^{i-1}$,
with $\piece{\tau}{\sigma}$ denoting the set of simplexes $\mu$ where $\tau \subseteq \mu \subseteq \sigma$.
\end{definition}

\begin{definition}
\label{definition-mapconnected}
A \emph{$q$-connected} carrier map $\Phi: \cK \to 2^{\cL}$ is
a strict carrier map such that, for all $\sigma \in \cK$,
$\dim(\Phi(\sigma)) > q - \codim_{\cK}(\sigma)$
and
$\Phi(\sigma)$ is $(q - \codim_{\cK}(\sigma))$-connected.
\end{definition}

\begin{definition}
\label{definition-mapshellable}
A \emph{$q$-shellable} carrier map $\Phi: \cK \to 2^{\cL}$ is
a strict carrier map such that, for all $\sigma \in \cK$,
$\dim(\Phi(\sigma)) > q - \codim_{\cK}(\sigma)$
and
$\Phi(\sigma)$ is shellable.
\end{definition}

After $r$ rounds,
note that $\cK^r$ only contains simplexes with dimension exactly $n -rk$.
In~\cite{ConcurrentShellable,MauriceBook},
the following lemmas are proved:

\begin{lemma}
\label{lemma-crashshellable}
For $1 \le i \le r$,
the operator $\cC^i: \cK^{i-1} \to 2^{\cK^i}$ is a $(k - 1)$-shellable carrier map.
\end{lemma}

\begin{lemma}
\label{lemma-compositionshellable}
If $\cM^1, \ldots, \cM^x$ are all $q$-shellable carrier maps,
and $\cM^{x+1}$ is a $q$-connected carrier map,
the composition $\cM^1 \circ \ldots \cM^x \circ \cM^{x+1}$ is a $q$-connected carrier map,
for any $x \ge 0$.
\end{lemma}

\textbf{Equivocation and Interpretation. }
After the crash-failure rounds,
if $m > 0$
the adversary picks one of the remaining processes to behave maliciously at round $r + 1$.
This process, say $P_b$, may send different views to different processes
(which is technically called \emph{equivocation}),
but,
informally speaking,
all views are ``plausible.''
For example,
two non-faulty processes $P_i$ and $P_j$
could be indecisive after round $r$ on whether the global state
is $\sigma_1$ or $\sigma_2$ in $\cK^r$,
while $P_b$,
a Byzantine process,
sends a state corresponding to $\sigma_1$ to $P_i$,
and a state corresponding to $\sigma_2$ to $P_j$.
The faulty process $P_b$ \emph{does not reveal} its Byzantine nature,
yet it \emph{promotes ambiguity} in the state information diffusion.

At the final round,
when a non-faulty process receives the states sent from the other processes,
it must decide correctly even if one other process equivocates.
If the non-faulty process can receive simplexes $\sigma_1$ and $\sigma_2$,
representing global states that differ in only one process's contribution
(that is, $\dim(\sigma_1 \cap \sigma_2) = n -rk - 1$),
then the \emph{interpretation} of a message containing one such state
must be the same as a message containing the other.
We capture this notion using the \emph{equivocation} operator,
called $\cE$,
describing the behavior of a Byzantine process,
coupled with an \emph{interpretation} operator,
called $\interp$,
describing the required behavior of non-faulty processes.
Informally,
$\interp(\sigma_1) = \interp(\sigma_2)$
for processes in $\names(\tau)$,
where $\tau = \sigma_1 \cap \sigma_2$ with $\dim(\tau) = n -rk - 1$.
Formally:

\begin{definition}
\label{definition-interp}
For any simplexes $\sigma_1$ and $\sigma_2$ in $\cK$,
with $\dim(\cK) = n - rk$,
let $(P_i, \interp(\sigma_1)) = (P_i, \interp(\sigma_2))$ if and only if
$\sigma_1 = \sigma_2$; \textbf{or}
$P_i \in \names(\tau)$ where $\tau = \sigma_1 \cap \sigma_2$ and $\dim(\tau) = n -rk - 1$.
\end{definition}

\begin{definition}
\label{definition-roundop-equivocate}
For any pure simplicial complexes $\cK$ and $\cL$
with $\dim(\cK) \le n - rk$ and $\cK \supseteq \cL$,
the $\cK$-equivocation operator $\cE_{\cK}$ is
\begin{align}
\label{eq-roundop-equivocate}
\cE_{\cK}(\cL) & = \bigcup_{\tau \in \Faces^{n-rk-1}(\cL)}
		             \Psi(\names(\tau); \set{\interp(\sigma^*): \sigma^* \in \cK, \sigma^* \supset \tau})) \mathrm{.}
\end{align}
\end{definition}
Note that $\cE_{\cK}(\cL) = \emptyset$ whenever $\dim(\cL) < n -rk - 1$ or $\dim(\cK) < n -rk$,
and also that
\begin{equation}
\cE_{\cK}(\sigma) = \bigcup_{\tau \in \Faces^{n-rk-1}(\sigma)} \Psi(\names(\tau); \interp(\sigma))
\end{equation}
for any $\sigma \in \cK$ with $\dim(\sigma) = n -rk$.
For convenience of notation,
define $\cE_{\cK}(\cK) = \cE(\cK)$.

Next, we investigate some technical properties of these constructions
that allow us to prove that the final complex is $(k-1)$-connected.
\begin{lemma}
\label{lemma-eqv:carrier}
For any pure, shellable simplicial complex with $\dim(\cK) \le n - rk$,
the $\cK$-equivocation operator $\cE_{\cK}$ is a carrier map.
\end{lemma}
\begin{proof}
Let $\tau \subseteq \sigma \in \cK$.
We show that $\cE_{\cK}(\tau) \subseteq \cE_{\cK}(\sigma)$.
If $\dim(\tau) < n - rk - 1$ then $\cE_{\cK}(\tau) = \emptyset$
and $\cE_{\cK}(\tau) \subseteq \cE_{\cK}(\sigma)$ for any $\sigma \supseteq \tau \in \cK$.
Otherwise,
if $\dim(\tau) = \dim(\sigma)$ then $\tau = \sigma$ and $\cE_{\cK}(\tau) = \cE_{\cK}(\sigma)$,
as we assumed that $\sigma \supseteq \tau \in \cK$.
The remaining case is when
$\dim(\tau) = n - rk - 1$ and $\dim(\sigma) = n - rk$,
which makes $\cE_{\cK}(\tau) \subseteq \cE_{\cK}(\sigma)$
in light of Definition~\ref{definition-roundop-equivocate}.
\end{proof}

Let $(\cC^r \circ \cE)$ be the composite map such that $(\cC^r \circ \cE)(\sigma) = \cE_{\cC^r(\sigma)}(\cC^r(\sigma))$.
While,
for an arbitrary complex $\cK$,
$\cE_{\cK}$ is not a strict carrier map \emph{per se},
we show in the following lemmas that $(\cC^r \circ \cE)$ is a $(k - 1)$-connected carrier map.
Lemma~\ref{lemma-lasttwo:strict} shows that $(\cC^r \circ \cE)$ is a strict carrier map,
and Lemma~\ref{lemma-lasttwo:conn} shows that
for any $\sigma \in \cK^{r-1}$,
$(\cC^r \circ \cE)(\sigma)$ is $((k - 1) - \codim_{\cK^{r-1}}(\sigma))$-connected.

\begin{lemma}
\label{lemma-lasttwo:strict}
$(\cC^r \circ \cE)$ is a strict carrier map.
\end{lemma}
\begin{proof}
Consider $\sigma,\tau \in \cK^{r-1}$,
with $\cL = \cC^r(\sigma)$ and $\cM = \cC^r(\tau)$.
Both $\cL$ and $\cM$ are pure, shellable simplicial complexes with dimension $n - rk$
(Definition~\ref{definition-roundop-crash} and Lemma~\ref{lemma-crashshellable}).
Therefore,
both the $\cL$-equivocation and $\cM$-equivocation operators are well-defined.
Also, $\cC^r$ is a strict carrier map,
hence $\cL \cap \cM = \cC^r(\sigma) \cap \cC^r(\tau) = \cC^r(\sigma \cap \tau)$.
Note that $\cL \cap \cM = \cC^r(\sigma \cap \tau)$,
if not empty,
is a pure, shellable simplicial complex with dimension $n - rk$.
Therefore,
the $(\cL \cap \cM)$-equivocation operator is well-defined.

First, we show that $\cE(\cL) \cap \cE(\cM) \subseteq \cE(\cL \cap \cM)$,
which implies one direction of our equality:
\begin{equation*}
\cE(\cC^r(\sigma)) \cap \cE(\cC^r(\tau)) \subseteq
\cE(\cC^r(\sigma) \cap \cC^r(\tau)) =
\cE(\cC^r(\sigma \cap \tau)) \mathrm{.}
\end{equation*}
For clarity, let $F(\cK) = \Faces^{n-rk-1}(\cK)$. Then,
\begin{align*}
\cE(\cL) \cap \cE(\cM)
= \bigcup_{\mu \in F(\cL)} \cE_{\cL}(\mu)
    \quad \cap \quad
    \bigcup_{\nu \in F(\cM)} \cE_{\cM}(\nu)
= \bigcup_{\substack{\mu \in F(\cL) \\ \nu \in F(\cM)}}
	\cE_{\cL}(\mu) \cap \cE_{\cM}(\nu) \mathrm{.}
\end{align*}
For arbitrary $\mu \in F(\cL)$ and $\nu \in F(\cM)$,
if $\cE_{\cL}(\mu) \cap \cE_{\cM}(\nu) \ne \emptyset$, consider two cases:
\begin{enumerate}
\item
$\mu$ and $\nu$ are proper faces of $\phi \in (\cL \cap \cM)$.
In this case,
$$  \cE_{\cL}(\mu) \cap \cE_{\cM}(\nu) = \Psi(\names(\mu) \cap \names(\nu); \interp(\phi)) \mathrm{,} $$
which is inside $\cE_{\cL \cap \cM}(\phi) \subseteq \cE_{\cL \cap \cM}(\cL \cap \cM)$.

\item
Otherwise, $\mu \subset \phi_1 \in \cL$ or $\nu \subset \phi_2 \in \cM$.
In this case,
$$ \cE_{\cL}(\mu) \cap \cE_{\cM}(\nu) = \Psi(\names(\mu) \cap \names(\nu); {\interp(\phi_1) \cap \interp(\phi_2)}) \mathrm{.} $$
By Definition~\ref{definition-interp},
the above is non-empty only when
$\interp(\phi_1) = \interp(\alpha)$ with $\alpha \in \cL$,
$\interp(\phi_2) = \interp(\beta)$ with $\beta \in \cM$,
and there exists a non-empty set $\bbP'$
such that $\bbP' \subseteq \names(\mu) \cap \names(\nu) \subseteq \names(\gamma)$,
where $\gamma = \alpha \cap \beta$ with $\dim(\gamma) = n -rk -1$.
Let $\bbP''$ be a maximal $\bbP'$ satisfying such condition.
Note that $\gamma \in (\cL \cap \cM)$, so $(\cL \cap \cM) \ne \emptyset$.

Since $(\cL \cap \cM)$ is non-empty,
it is pure, shellable with dimension $n -rk$,
there must exist a simplex $\gamma' \supset \gamma$ with dimension $n -rk$.
Moreover,
$\interp(\gamma') = \interp(\alpha) = \interp(\phi_1)$
and
$\interp(\gamma') = \interp(\beta) = \interp(\phi_2)$
for processes in $\names(\gamma)$,
given the definition of $\interp$.
In conclusion, we have
$\cE_{\cL}(\mu) \cap \cE_{\cM}(\nu) = \Psi(\bbP''; \interp(\gamma')) \subseteq \Psi(\names(\gamma); \interp(\gamma'))$,
which is inside $\cE_{\cL \cap \cM}(\gamma') \subseteq \cE_{\cL \cap \cM}(\cL \cap \cM)$.
\end{enumerate}

In the other direction,
we have
$\cE(\cL \cap \cM) \stackrel{\textrm{def}}{=} \cE_{\cL \cap \cM}(\cL \cap \cM) \subseteq \cE_{\cL}(\cL \cap \cM)
\subseteq \cE_{\cL}(\cL) \stackrel{\textrm{def}}{=} \cE(\cL)$,
since
(i) $\cE_{\cL \cap \cM}(\cX) \subseteq \cE_{\cL}(\cX)$ for any $\cX \subseteq \cL \cap \cM$
(Definition~\ref{definition-roundop-equivocate});
and (ii)
$\cE_{\cL}$ is a carrier map (Lemma~\ref{lemma-eqv:carrier}).
The same argument proves that $\cE(\cL \cap \cM) \subseteq \cE(\cM)$,
and therefore $\cE(\cL \cap \cM) \subseteq \cE(\cL) \cap \cE(\cM)$.
\end{proof}

\begin{lemma}
\label{lemma-lasttwo:conn}
For any $\sigma \in \cK^{r-1}$,
$\cE(\cC^r(\sigma))$ is
$((k - 1) - \codim_{\cK^{r-1}}(\sigma))$-connected.
\end{lemma}
\begin{proof}
Consider $\sigma \in \cK^{r-1}$ with $\codim_{\cK^{r-1}}(\sigma) \le k$.
By Lemma~\ref{lemma-crashshellable},
$\cM = \cC^r(\sigma)$ is a pure, shellable simplicial complex with $\dim(\cM) = n - rk = d$.
By Definition~\ref{definition-roundop-equivocate},
$\cE(\cM)$ is well-defined and
$\dim(\cE(\cM)) = n - rk - 1 = d'$.
Note that $d' \ge n - t \ge 2t \ge 2k$,
since $n + 1 > 3t$ and $t \ge k$.

First,
we show that $\cE(\cM)$ is ``highly-connected'' -- that is, $(2k - 1)$-connected.
We proceed by induction on $\mu_0 \ldots \mu_{\ell}$,
a shelling order of facets of $\cM$.
\begin{description}
	\item[Base.] We show that $\cE_{\cM}(\mu_0)$ is $(2k - 1)$-connected.
Considering Definition~\ref{definition-roundop-equivocate},
we have that $\cE_{\cM}(\mu_0) = \cE_{\cM}(\tau_0) \cup \ldots \cup \cE_{\cM}(\tau_d)$,
with $\tau_0 \ldots \tau_d$ being all the proper faces of $\mu_0$.

Consider the cover $\set{\cE_{\cM}(\tau_i): 0 \le i \le d}$ of $\cE_{\cM}(\mu_0)$,
and its associated nerve $\cN(\set{\cE_{\cM}(\tau_i): 0 \le i \le d})$.
For any index set $J \subseteq I = \set{0 \ldots d}$,
let
\begin{align*}
\cK_J =
\bigcap_{j \in J} \cE_{\cM}(\tau_j) = \Psi(\bigcap_{j \in J} \names(\tau_j); \interp(\mu_0))
\end{align*}
For any $J$ with $|J| \le d$,
we have $\cap_{j \in J} \names(\tau_j) \ne \emptyset$,
making $\cK_J$ a non-empty pseudosphere
with dimension $d' - |J| + 1 \ge 2k - |J| + 1$.
So, $\cK_J$ is $((2k - 1) -  |J| + 1)$-connected
by Lemmas~\ref{lemma-pseudosphere-shellable} and~\ref{lemma-shellable-connected}.
The nerve is hence the $(d - 1)$-skeleton of $I$,
which is $(d - 2) = (d' - 1) \ge (2k - 1)$-connected.
By the Nerve Theorem,
$\cE_{\cM}(\mu_0)$ is also $(2k - 1)$-connected.

	\item[IH.] Assume that $\cY = \cup_{0 \le y < x} \cE_{\cM}(\mu_y)$ is $(2k - 1)$ connected,
and let $\cX = \cE_{\cM}(\mu_x)$.
We must show that $\cY \cup \cX = \cup_{0 \le y \le x} \cE(\mu_y)$ is $(2k - 1)$-connected.
Note that $\cX$ is $(2k - 1)$-connected by an argument identical to the one above for the base case $\cE_{\cM}(\mu_0)$.
Besides,
\begin{align*}
\cY \cap \cX
= \left( \bigcup_{0 \le y < x} \cE_{\cM}(\mu_y) \right) \cap \cE_{\cM}(\mu_x)
= \bigcup_{0 \le y < x} (\cE_{\cM}(\mu_y) \cap \cE_{\cM}(\mu_x))
\stackrel{\star}{=} \bigcup_{i \in S} \cE_{\cM}(\tau_i) \mathrm{,}
\end{align*}
where $i \in S$ is such that $(\cup_{0 \le y < x} \, \mu_y) \cap \mu_x = \cup_{i \in S} \, \tau_i$.
The set $S$ is well-defined since $\cM$ is shellable.
The step $(\star)$ holds because:
(i) $\cY \cap \cX$ must include at least $\bigcup_{i \in S} \cE_{\cM}(\tau_i)$;
and (ii) $\cE_{\cM}(\mu_y) \cap \cE_{\cM}(\mu_x) \ne \emptyset$ only if
$\psi = \Psi(\names(\mu_y \cap \mu_x); \interp(\mu_x))$ exists,
the latter inside $\psi' = \Psi(\names(\tau_j); \interp(\mu_x))$
for some $j \in S$,
or we contradict the fact that $\cM$ is shellable.

Using an argument identical to the one for $\cE_{\cM}(\mu_0)$,
yet considering the cover $\set{\cE_{\cM}(\tau_i): i \in S}$,
the nerve of $\cX \cap \cY$ is either the $(d - 1)$-skeleton of $S$ (if $S = \set{0 \ldots d}$) or the whole simplex $S$ (otherwise).
By the Nerve Theorem,
$\cup_{i \in S} \cE_{\cM}(\tau_i)$ is $(2k - 1)$-connected.

Once again,
using the Nerve Theorem,
since $\cY$ is $(2k - 1)$-connected,
$\cX$ is $(2k - 1)$-connected,
and $\cY \cap \cX$ is $(2k - 1)$-connected,
we have that $\cY \cup \cX$ is $(2k - 1)$-connected.
\end{description}
While the equivocation operator yields high connectivity ($2k - 1$)
in the pseudosphere $\cC^r(\sigma)$,
the \emph{composition} of $\cC^r$ and $\cE_{\cC^r(\sigma)}(\cC^r(\sigma))$
limits the connectivity to $(k - 1)$,
since the former map is only defined for simplexes with codimension $\le k$.
Formally,
as $\cC^r(\sigma) \ne \emptyset$ for any simplex $\sigma \in \cK^{r-1}$ with $\codim_{\cK^{r-1}}(\sigma) \le k$,
we have that $\cE(\cC^r(\sigma))$ is $((k - 1) - \codim_{\cK^{r-1}}(\sigma))$-connected.
\end{proof}

From Lemmas~\ref{lemma-lasttwo:strict} and~\ref{lemma-lasttwo:conn},
we conclude the following.
\begin{corollary}
\label{corollary-lasttwo:conn}
$(\cC^r \circ \cE)$ is a $(k - 1)$-connected carrier map.
\end{corollary}

\begin{theorem}
An adversary can keep the protocol complex of a Byzantine synchronous system $(k - 1)$-connected
for $\lceil t/k \rceil$ rounds.
\end{theorem}
\begin{proof}
If $m = 0$,
$t \bmod k = 0$,
and the adversary runs only the crash rounds failing $k$ processes each time,
for $r = \lfloor t/k \rfloor = \lceil t/k \rceil$ consecutive rounds.
We have the following scenario:
\begin{displaymath}
(\cC^1 \circ \ldots \circ \cC^r)(\sigma) \mathrm{.}
\end{displaymath}
Since $\cC^i: \cK^{i-1} \to 2^{\cK^i}$ is a $(k - 1)$-shellable carrier map
for $1 \le i \le r$ (Lemma~\ref{lemma-crashshellable}),
the composition $(\cC^1 \circ \ldots \circ \cC^r)$
is a $(k - 1)$-connected carrier map for any facet $\sigma \in \cI$
(Lemma~\ref{lemma-compositionshellable}).

If $m > 0$,
the adversary performs $r$ crash rounds (failing $k$ processes each time),
followed by the extra equivocation round.
We have the following scenario:
\begin{equation}
\label{equation-equivrun}
(\cC^1 \circ \ldots \circ \cC^{r-1} \circ (\cC^r \circ \cE))(\sigma)
\mathrm{.}
\end{equation}
Since $\cC^i: \cK^{i-1} \to \cK^i$ is a $(k - 1)$-shellable carrier map
for $1 \le i \le r - 1$ (Lemma~\ref{lemma-crashshellable}),
and $(\cC^r \circ \cE)$ is a $(k - 1)$-connected carrier map
(Corollary~\ref{corollary-lasttwo:conn}),
we have that the composition above $(\cC^1 \circ \ldots \circ \cC^{r-1} \circ (\cC^r \circ \cE))$
is a $(k - 1)$-connected carrier map for any facet $\sigma \in \cI$
(Lemma~\ref{lemma-compositionshellable}).
\end{proof}

\section{$k$-Set Agreement and Lower Bound}
\label{Sec-KSetLowerBound}

The $k$-set agreement problem and connectivity are closely related.
Lemma~\ref{lemma-connectivitykset},
proved in Appendix~\ref{App-ProofsConnectivity},
shows that
no solution is possible for $k$-set agreement
with a $(k - 1)$-connected protocol complex,
which,
as seen in Sec.~\ref{Sec-ConnectivityUpperBound},
can occur at least until round $\lceil t/k \rceil$.
\begin{lemma}
\label{lemma-connectivitykset}
If, starting $\sigma \in \cI$,
the protocol complex $\cP(\sigma)$ is $(k - 1)$-connected,
then no decision function $\delta$ solves the $k$-set agreement problem.
\end{lemma}

We now present a simple $k$-set agreement algorithm
for Byzantine synchronous systems,
running in $\lceil t/k \rceil + 1$ rounds.
The procedure requires a relatively large number of processes compared to $t$:
we assume $n + 1 \ge k(3t + 1)$.
The procedure was designed 
with the purpose of tightening the connectivity lower bound,
favoring simplicity over the optimality on the number of processes.

Non-faulty processes initially execute a \emph{gossip phase}
for $\lceil t/k \rceil + 1$ rounds,
followed by a \emph{validation phase},
and a \emph{decision phase},
where the output is chosen.
Define $R = \lceil t/k \rceil$,
and consider the following tree,
where nodes are labeled with words over the alphabet $\bbP$.
The root node is labeled as $\lambda$,
which represents an empty string.
Each node $w$ such that $0 \le |w| \le R$
has $n + 1$ child nodes labeled $wp$ for all $p \in \bbP$.
Any non-faulty process $P_i$ maintains such tree,
denoted $T_i$.

All nodes $w$ are associated with the value $\cont_p(w)$,
called the \emph{contents} of $w$.
The special value $\bot$ represents an absent input.
We omit the subscript $p$ when the process is implied or arbitrary.
We divide the processes into $k$ disjoint groups: $\bbP(g) = \set{P_x \in \bbP: x = g \bmod k}$,
for $0 \le g < k$.
For any tree $T$,
we call $T(g)$ the subtree of $T$ having only nodes $wp \in T$ such that $p \in \bbP(g)$.

\begin{algorithm}[htb]
\caption{$P_x.\mathrm{Agree}(I)$}
\label{Alg-Agree}
\begin{algorithmic}[1]
\If{$k = 1$}
	\State \Return $\mathtt{Decision}(\mathrm{Multiset}(\cont(p): p \textrm{ output by consensus algorithm}))$
\EndIf
\State $\cont(w) \leftarrow \bot$ for all $w \in T$
\State $\cont(\lambda) \leftarrow I$ \Comment{Gossip}
\For{$\ell: 1 \textbf{ to } \lceil t/k \rceil + 1$}
	\State $\send{S_x^{\ell-1} = \set{(w, \cont(w)): |w| = \ell - 1}}$
	\Upon{$\recv{S_y^{\ell-1} = \set{(w, v): |w| = \ell - 1, v \in V \cup \set{\bot}}}$ from $P_y$}
		\State $\cont(w P_y) \leftarrow v$ for all $(w, v) \in S_y^{\ell-1}$
	\EndUpon
\EndFor \label{algAgree:gossip}
\State $\bbP' \leftarrow \set{P_i: P_i \textrm{ has a quorum}}$ \Comment{Validation}
\If{$|\bbP'| = (n + 1) - t$}
	\State Apply completion rule for all $wb$ where $b \in \bbP \setminus \bbP'$ and $|wb| = \lceil t/k \rceil$
\EndIf
\State $g \leftarrow \ $ any $g$ such that $T(g)$ is pivotal \label{algAgree:dec1} \Comment{Decision}
\For{$\ell: \lceil t/k \rceil - 1 \textbf{ to } 1$}
	\State Apply consensus rule for all non-validated $wb$ where $b \in \bbP(g)$ and $|wb| = \ell$ \label{algAgree:dec2}
\EndFor
\State \Return $\mathtt{Decision}(\mathrm{Multiset}(\cont(p): p \in T(g)))$
\end{algorithmic}
\end{algorithm}

In the validation phase,
if we have a set $\bbQ$ containing $(n + 1) - t$ processes that acknowledge
all messages transmitted by process $p$ (making sure that $p \in \bbQ$),
at every round $1 \le r \le R$,
we call such set the \emph{quorum} of $p$,
denoted $\quorum{p}$.
Formally,
$\quorum{p} = \bbQ \subseteq \bbP$
such that $p \in \bbQ$,
$|\bbQ| \ge (n + 1) - t$,
and $q \in \bbQ$ whenever $\cont(wp) = v$ implies $\cont(wpq) = v$,
for any $wp$ with $0 \le |wp| \le R$.
It should be clear that every non-faulty process has a quorum containing at least
all other non-faulty processes.
If a process $p$ has a quorum as seen by process $P_i \in \bbG$,
we say that $wp$ has been \emph{validated} on $P_i$,
for any $wp$ with $0 \le |wp| < R$
(and that $p$ has been validated on $P_i$).
Note that in our definition either all entries $wp$ for $p \in \bbP$ are validated,
or none is.
Lemma~\ref{lemma-validated-identical},
proven in Appendix~\ref{App-ProofsAlgorithm},
shows that validated entries are unique across non-faulty processes.
\begin{lemma}
\label{lemma-validated-identical}
If $p$ has been validated on non-faulty processes $P_i$ and $P_j$,
then $\cont_i(wp) = \cont_j(wp)$ for any $0 \le |w| < R$.
\end{lemma}

In the decision phase,
if we see $t$ processes without a quorum,
we have technically identified all non-faulty processes $\bbB$.
In this case,
we fill $R$-th round values of any $b \in \bbB$ using the \emph{completion} rule:
we make $\cont(wb) = v$ if we have $(n + 1) - 2t$ processes $\bbG' \subseteq \bbG$
where $\cont(wbg) = v$
for any $g \in \bbG'$ and $|wb| = R$.
If a process $b$ has its $R$-round values
completed as above in process $P_i \in \bbG$,
we say that $wb$ has been \emph{completed} on $P_i$ for any $|wb| = R$.
Lemma~\ref{lemma-validated-completed},
proven in Appendix~\ref{App-ProofsAlgorithm},
shows that completed entries are identical and consistent
with validated entries across non-faulty processes.
(Intuitively,
the completion rule was done over identical values from correct processes.)
\begin{lemma}
\label{lemma-validated-completed}
If $wp$ has been completed or validated on a non-faulty process $P_i$,
and $wp$ has been completed on a non-faulty process $P_j$,
then $\cont_i(wp) = \cont_j(wp)$.
\end{lemma}

We have two possible cases:
(i) there is a subtree $T(g)$ with less than $\lceil t/k \rceil$
non-validated processes --
call such subtree \emph{pivotal};
or (ii) no such tree exists,
in which case we apply the completion rule to $R$-round values in $T(0)$,
and define $T(0)$ as our pivotal subtree instead.
A pivotal subtree, therefore, must exist according to the definition above.

Denote the set of processes in the word $w$ as $\setproc{w}$.
For any non-validated $wb$ with $b \in \bbP(g)$ in a pivotal subtree $T(g)$,
where $1 \le |wb| < R$,
we establish consensus on $\cont(wb)$.
We apply the \emph{consensus} rule:
$\cont(wb) = v$ if the majority of processes in $\bbP(g) \setminus \setproc{wb}$
is such that $wbp = v$.
This rule is applied first to entries labeled $wb$ where $|wb| = R - 1$,
and then moving upwards (please refer to Alg.~\ref{Alg-Agree}).
Our algorithm is essentially
separating the possible chains of unknown values across disjoint process groups,
which either forces one of these chains to be smaller than $R = \lceil t/k \rceil$,
or reveals all faulty processes,
giving us the ability to perform the completion rule.
This fundamental tradeoff underlies our algorithm,
and ultimately explains \emph{why} the $\lceil t/k \rceil$ connectivity bound is tight.
Lemma~\ref{lemma-consensus-pivotal},
proven in Appendix~\ref{App-ProofsAlgorithm},
shows that the consensus rule indeed establishes consensus
across non-faulty processes that identify $T(g)$ as the pivotal subtree.

\begin{lemma}
\label{lemma-consensus-pivotal}
For any two-non-faulty processes $P_i$ and $P_j$ that applied the consensus rule on a pivotal subtree $T(g)$,
with $0 \le g < k$,
we have that $\cont_i(p) = \cont_j(p)$ for any $p \in \bbP(g)$.
\end{lemma}

\begin{theorem}
\label{theorem-ksetcorrectness}
Algorithm~\ref{Alg-Agree} solves $k$-set agreement
in $\lceil t/k \rceil + 1$ rounds.
\end{theorem}
\begin{proof}
Termination is trivial,
as we execute exactly $R = \lceil t/k \rceil + 1$ rounds. 
By Lemma~\ref{lemma-consensus-pivotal},
each pivotal subtree yields a unique decision value.
As we have at most $k$ pivotal subtrees identified across non-faulty processes,
up to $k$ values are possibly decided across non-faulty processes.
\end{proof}

\section{Conclusion}

In Byzantine synchronous systems,
the protocol complex can remain $(k - 1)$-connected for $\lceil t/k \rceil$ rounds,
potentially \emph{one} more round than in crash-failure systems.
We conceive a combinatorial operator modeling the ability of Byzantine processes
to equivocate without revealing their Byzantine nature,
just after $\lfloor t/k \rfloor$ rounds of crash failures.
We compose this operator with the regular crash-failure operators,
extending $(k - 1)$-connectivity up to $\lceil t/k \rceil$ rounds.
We tighten this bound,
at least when $n$ is relatively large compared to $t$,
via a full-information protocol
that solves a formulation of $k$-set agreement.

It may be surprising that Byzantine failures
impose only \emph{one} additional synchronous round over the crash-failure model,
and \emph{at most that} in our standard setting,
where inputs are arbitrarily attributed to processes,
and the number of processes is strictly bigger than $k(3t + 1)$.
In terms of solvability vs. number of rounds,
the penalty for moving from crash to Byzantine failures
can thus be \emph{quite limited}.
Previous work has hinted this possibility operationally, since
(i) in synchronous systems where $n$ is large enough compared to $t$,
we can simulate crash failures on Byzantine systems with a 1-round delay~\cite{BazziNeiger01};
and (ii) techniques similar to the reliable broadcast of~\cite{Bracha,SriTouRB}
deal with the problem of Byzantine equivocation,
also with a 1-round delay.
This extra round is crucial -- but enough -- to limit the impact of Byzantine behavior
in rather usual operational settings.


\newpage

\appendix

\section{Appendix: Proofs for the Connectivity Arguments}
\label{App-ProofsConnectivity}

\noindent \textbf{Proof of Lemma~\ref{lemma-connectivitykset}}
\begin{proof}
Consider a $k$-simplex $\alpha = \set{u_0, \ldots, u_k} \subseteq \set{v_0, \ldots, v_d}$ with $k + 1$ different inputs.
Let $\cI_\beta = \Psi(\bbP, \beta)$ for any $\beta \subseteq \alpha$,
and $\cI_x = \bigcup_{\beta \in \skel^x(\alpha)} \Psi(\bbP, \beta)$.
We construct a sequence of continuous maps $g_x: |\skel^x(\alpha)| \to |\cK_x|$
where $\cK^x$ is homeomorphic to $\skel^x(\alpha)$ in $|\skel^x(\cP(\cI_x))|$.

\textbf{Base.} Let $g_0$ map any vertex $v \in \alpha$ to a vertex in $\cK_v = \cP(\cI_{\set{v}})$.
We know that $\cK_v$ is $k$-connected since $\dim(\cI_{\set{v}}) = \dim(\cI)$ and $\cP$ is a $k$-connected carrier map.
We just constructed
$$ g_0: |\skel^0(\alpha)| \to |\cK_0| \mathrm{,} $$
where $\cK^0$ is isomorphic to a $\skel^0(\alpha)$ in $|\skel^0(\cP(\cI_0))|$.

\textbf{Induction Hypothesis.} Assume $g_{x-1}: |\skel^{x-1}(\alpha)| \to |\cK_{x-1}|$ for any $x \le k$,
where $\cK_{x-1}$ is isomorphic to $\skel^{x-1}(\alpha)$ in $|\skel^{x-1}(\cP(\cI_{x-1}))|$. 
For any $\beta \in \skel^x(\alpha)$,
we have that $\skel^x(\cP(\cI_\beta))$ is $(x - 1)$-connected,
hence the continuous image of the $(x - 1)$-sphere in $\cP(\cI_\beta)$
can be extended to the continuous image of the $x$-disk in $\skel^x(\cP(\cI_\beta))$.
We just constructed
$$ g_x: |\skel^x(\alpha)| \to |\cK_x| \mathrm{,} $$
where $\cK^x$ is isomorphic to $\skel^x(\alpha)$ in $|\skel^x(\cP(\cI_0))|$.
In the end,
we have $g_k: |\alpha| \to |\cK_k|$ where $\cK_k$ is isomorphic to $\alpha$ in $\skel^k(\cP(\cI_k))$.

Now suppose,
for the sake of contradiction,
that $k$-set agreement is solvable,
so there must be a simplicial map $\delta: \cP(\cI) \to \cO$ carried by $\Delta$.
Then,
induce the continuous map $\delta_c: |\cK_k| \to |\alpha|$ from $\delta$ such that
$\delta_c(v) \in |\views(\delta(\mu))|$ if $v \in |\mu|$,
for any $\mu \in \cK_k$.
Also,
note that the composition of $g_k$ with the continuous map $\delta_c$
induces another continuous map $|\alpha| \to |\Bd\alpha|$,
since by assumption $\delta$ never maps a $k$-simplex of $\cK_k$ to a simplex with $k + 1$ different views
(so $\delta_c$ never maps a point to $|\inte{\alpha}|$).
We built a \emph{continuous retraction} of $\alpha$ to its own border $\Bd\alpha$,
a contradiction (please refer to~\cite{Munkres84,Kozlov07}).
Since our assumption was that there existed
a simplicial map $\delta: \cP(\cI) \to \cO$ carried by $\Delta$,
we conclude that $k$-set agreement is not solvable.
\end{proof}

\section{Appendix: Proofs for the $k$-Set Agreement Procedure}
\label{App-ProofsAlgorithm}

\noindent \textbf{Proof of Lemma~\ref{lemma-validated-identical}.}
\begin{proof}
If $p$ has been validated on $P_i \in \bbG$,
then $\cont_i(wp) = v$ implies $\cont_i(wpq) = v$ for $(n + 1) - t$ different processes $q \in \bbQ_i$,
and $\cont_j(wp) = v$ implies $\cont_j(wpq) = v$ for $(n + 1) - t$ different processes $q \in \bbQ_j$,
for any $0 \le |w| < R$.
As we have at most $t$ non-faulty processes and $n + 1 > 3t$,
$|\bbQ_i \cap \bbQ_j| \ge (n + 1) - 2t > t + 1$,
containing at least one non-faulty process that,
by definition,
broadcasts values consistently in its run.
Hence, $\cont_i(wp)$ and $\cont_j(wp)$ must be identical.
\end{proof}

\noindent \textbf{Proof of Lemma~\ref{lemma-validated-completed}}
\begin{proof}
If $wp$ has been validated on $P_i$,
$\cont_i(wp) = v$ implies $\cont_i(wpq) = v$ for $(n + 1) - t$ different processes $q \in \bbQ$.
When $P_j$ applies the completion rule on $wp$,
then $\cont_j(wpq) = v$ for $(n + 1) - 2t$ different processes $q \in \bbG$,
as we have at most $t$ faulty processes.
Therefore,
$\cont_i(wp) = \cont_j(wp)$.

If $wp$ has been completed on all non-faulty processes,
they all have identified $t$ faulty processes,
and the completion rule is performed over identical entries associated with non-faulty processes.
Therefore,
$\cont_i(wp) = \cont_j(wp)$ as well.
\end{proof}

\noindent \textbf{Proof of Lemma~\ref{lemma-consensus-pivotal}}
\begin{proof}
Consider a non-faulty process $P_i$ establishing the value of $\cont_i(wp)$
with the consensus rule.
Define $\setcons{wp} = \bbP(g) \setminus \setproc{wp}$ for any $wp \in T(g)$ with $|wp| < R$,
noting that $|\setcons{wp}| \ge 2t + 2$ as $|\bbP(g)| \ge 3t + 1$ and $|wp| < t$.

Consider two cases:
(i) if $wp$ has been validated at a non-faulty process $P_j$ with $\cont_j(wp) = v$,
at most $t$ values from $S_i = \mathrm{Multiset}(\cont_i(wpq): q \in \setcons{wp})$
will be different than $v$.
Hence, there will always be a majority of values in $S_i$ that will contain $v$,
because $|S_i| \ge 2t + 2$.
(ii) otherwise, if $wp$ has not been validated at any non-faulty process,
all $\cont(wp)$ values are being calculated over consistent values,
by Lemma~\ref{lemma-validated-completed},
which makes all non-faulty processes establish $\cont(wp)$ consistently with the consensus rule.
\end{proof}

\bibliography{Bibliography.bib}
\bibliographystyle{abbrv}

\end{document}